\documentclass[12pt]{article}

\usepackage{amsmath, amssymb}
\usepackage{graphicx}
\usepackage{enumerate}
\usepackage{natbib}
\usepackage{overpic}
\usepackage{subfiles}
\usepackage{hyperref}
\usepackage{cleveref}

\newcommand{\blind}{1}

\usepackage{enumitem}
\usepackage{mathrsfs, booktabs, amsthm, amsfonts, bbold}
\usepackage{accents}

\usepackage{xcolor}
\protected\def\[#1\]{\begin{equation}\begin{aligned}#1\end{aligned}\end{equation}}
\protected\def\(#1\){\begin{equation*}\begin{aligned}#1\end{aligned}\end{equation*}}

\newtheorem{theorem}{Theorem}

\theoremstyle{remark}
\newtheorem{remark}{Remark}

\theoremstyle{definition}

\newcommand{\diag}{\operatorname{diag}}

\newcommand{\dom}{\operatorname{dom}}

\newcommand{\dist}{\operatorname{dist}}

\usepackage{subcaption}
\usepackage{float}

\usepackage{comment}

\theoremstyle{definition}

\begin{document}

\def\spacingset#1{\renewcommand{\baselinestretch}
{#1}\small\normalsize} 
\spacingset{1}

\if1\blind
{
  \title{\bf Revisiting Bayesian Variable Selection\\ via Optimization}
  \author{
	Leo L Duan \thanks{\href{email:li.duan@ufl.edu}{li.duan@ufl.edu}}\\
	Department of Statistics, University of Florida    }
    \date{}
  \maketitle
  \vspace{-0.5in}
} \fi

\if0\blind
{
  \bigskip
  \bigskip
  \bigskip
  \begin{center}
    {\LARGE\bf }
\end{center}
  \medskip
} \fi

\bigskip

\begin{abstract}
    Variable selection in linear regression has been a central topic in statistical research for decades. Bayesian variable selection methods, which account for uncertainty in both the regression coefficients and the noise variance, have achieved broad success through the use of discrete or continuous shrinkage priors and efficient collapsed Gibbs samplers. Despite their popularity and strong empirical performance, an enigma remains: the marginal likelihood, obtained by integrating out the regression coefficients and noise variance, is not log-concave; therefore, there is no guarantee of reliably finding its global optimum. In this article, we study this problem from an optimization perspective. Taking the negative log-marginal likelihood as a loss function of the latent precision parameters, we can rewrite it as a difference of convex functions (DC), and then optimize it via a simple iterative algorithm. Under mild compact set conditions, the DC algorithm converges to the global optimum at a linear rate. The positive finding applies to type-II maximum likelihood and extends to maximum marginal posterior under suitable priors, indicating that the problem of mode finding in Bayesian variable selection is much more benign than the lack of log-concavity might suggest. Besides the theoretical insight, the proposed algorithm is easy to implement, free of tuning, and extensible to structured sparsity, and thus can serve as an efficient alternative or warm-start for traditional Markov chain Monte Carlo solutions. The method is illustrated through numerical studies and a spatial data application for quantifying the aftershock risk following the 2019 Ridgecrest earthquakes.

    The source code for the algorithm is publicly available at \url{https://github.com/leoduan/dca_optimization_variable_selection}.
\end{abstract}
\vfill

\addtolength{\textheight}{-0.8in}
\addtolength{\topmargin}{0.5in}

\newpage

\spacingset{1.8} 

\section{Introduction}

Variable selection for linear regression has long held a central place in statistical research, especially when the number of candidate predictors is comparable to or exceeds the sample size. In these scenarios, it is important to distinguish true signals from noise while properly quantifying uncertainty. Bayesian methods have proven particularly appealing, due to the convenient probabilistic framework for simultaneously performing model selection and inference on the non-zero coefficients.

In this article, we focus on the hierarchical model based on normal regression likelihood with conjugate normal prior and inverse-gamma noise variance prior:
\[\label{eq:model}
y \mid \theta,\sigma^2 \sim \text{N}(X\theta,\sigma^2 I_n),\qquad
\theta \mid \sigma^2,D \sim \text{N}(0,\sigma^2 D^{-1}),\qquad
\sigma^2 \sim \operatorname{IG}(a_0,b_0),
\]
where $X\in\mathbb R^{n\times p}$ is the design matrix and
$
D=\diag(d_1,\dots,d_p)\succ 0.$ Typically, one assigns another prior on $D$, leading to a normal-scale mixture prior on $\theta$.

Under this formulation, each coordinate $d_j$ controls the amount of shrinkage applied to $\theta_j$: larger values of $d_j$ correspond to smaller prior variance for $\theta_j$, and therefore stronger shrinkage of $\theta_j$ toward zero.
 The discrete spike-and-slab prior \citep{george1993} can be viewed as a limiting case when $d_j\to \infty$ for those $\theta_j=0$ in the spike, while $d_j<\infty$ for those $\theta_j\neq 0$ in the slab.
Relaxing from the degenerate form, one obtains a continuous shrinkage prior on $\theta$. For example, one may place a two-atom discrete prior so that $1/d_j=\delta_0\approx 0$ for the spike and $1/d_j=\delta_1\gg \delta_0$ for the slab, leading to the continuous spike-and-slab prior \citep{george1993}.
The framework is also naturally related to the recent literature on global--local shrinkage priors \citep{polson_scott_global_local}: for which, one imposes a global scale on $(1/d_j)$ to be small, while allowing heavy tail in the distribution of $(1/d_j)$ so that there is some positive probability for $|\theta_j|$ to be large, hence escaping excessive shrinkage. 
Prominent examples include the horseshoe prior \citep{carvalho2010}, 
the normal-exponential-gamma prior \citep{brown2010inference},
the generalized double Pareto prior \citep{armagan2013},
the Dirichlet--Laplace prior \citep{bhattacharya2015},
the beta-prime prior \citep{bai2021}. Related constructions also connect continuous shrinkage with a two-atom mixture prior, as in the spike-and-slab lasso \citep{rockova2018}.

In the posterior computation under most of the above priors, the normal-scale mixture representation in \eqref{eq:model} is amenable to efficient Markov chain Monte Carlo (MCMC) algorithms, namely, collapsed Gibbs sampling: one integrates out $\theta$ and $\sigma^2$ to obtain a marginal likelihood $\mathcal L(y\mid D)$; with prior $\pi_0(D)$, one could update $D$ in a block or by individual coordinate (and update the other hyperparameters if exist). After collecting the marginal chain of $D$, one then samples the conditional posterior of $(\theta,\sigma^2)$  given $D$. Such a strategy and its variants have been widely adopted, see discussions in \cite{george1993,rajaratnam2015scalable,bhattacharya2016fast} including approximate algorithms \citep{johndrow2020scalable} and extensions to generalized linear models via data augmentation \citep{albert1993bayesian,polson2013bayesian}.

Despite the success of these methods, an important  question remains unanswered: the marginal likelihood $\mathcal L(y\mid D)$ is not log-concave (neither is its product with the prior $\pi_0(D)$ for most of the above priors), so doubts arise on the reliability of finding the global optimum, that is, the maximum marginal likelihood estimate or posterior mode for $D$. 

In this article, we formally revisit this question via the lens of non-convex optimization. We show that the negative log-marginal likelihood admits a useful difference-of-convex structure. Under mild conditions, the global optimum can be attained rapidly by a simple iterative algorithm, regardless of the initial value. Our analysis of the optimization problem shows that the exploration towards the global optimum neighborhood is substantially more benign than the lack of log-concavity alone might suggest, shedding new light on the effectiveness of Bayesian variable selection methods. Besides the theoretical insight, the proposed algorithm is simple and tuning-free, and enjoys great extensibility for inducing structured sparsity \citep{griffin2024structured}.

Before we discuss the focused problem, we want to distinguish our method from three major groups of literature. First, our optimization solution operates on the marginal likelihood, which
can be viewed as a type-II maximum likelihood approach \citep{pena2020restricted}, and is distinct from the early work on optimizing over the regression coefficients \citep{tibshirani1996, fanli2001, zhang2010nearly}. Since the marginal likelihood incorporates the uncertainty of both the regression coefficients and the noise variance, it quantifies the uncertainty of the variable selection decision for statistical inference.
Second, there is a MCMC convergence theory literature supporting the convergence of Gibbs sampling Markov chains to the target posterior, under certain variable selection priors \citep{JohnsonJones2010, RomanHobert2015, Vats2017}. In other words, the literature ensures after a sufficient number of MCMC iterations, a newly drawn Markov chain sample is distributionally indistinguishable from one from the target posterior.
In contrast, our focus here is on the highest mode of the posterior itself, a fundamental aspect of the target posterior. Lastly, our considered marginal likelihood is closely related to the partially integrated likelihood (with $\theta$ integrated out, but not $\sigma^2$) studied in the relevance vector machine literature \citep{tipping2001sparse}, which first discovered an automatic variable selection mechanism when $D$ is maximized over. 
Later, \cite{faul2001analysis} established the uniqueness of coordinate-wise maximizer of $d_j$ given the other $d_{k},k\neq j$. With the above connections, it is important to emphasize that for non-convex problems, coordinate-wise optimality does not ensure global optimality, nor does it guarantee iterative coordinate-wise updates will converge to the global optimum. In that regard, our work fills the critical gap, and the theory and algorithm extend to mode finding of the marginal posterior when certain priors are imposed on $D$. We illustrate the practical performance advantage of our algorithm over existing optimization-based methods in the numerical studies.

\section{Marginal likelihood and difference-of-convex optimization}

\subsection{Variable selection via optimization of marginal likelihood}
In this subsection, we first present the form of the loss function, and carefully reinterpret a few existing results from the literature through the lens of optimization.
Focusing on the model \eqref{eq:model}, after integrating out $\theta$ and $\sigma^2$, we obtain the marginal likelihood dependent on the precision parameters $D$:
\(
\mathcal L(y\mid D)
\propto
|I_n+XD^{-1}X^\top|^{-1/2}
\left(
b_0+\frac12\,y^\top(I_n+XD^{-1}X^\top)^{-1}y
\right)^{-(a_0+\frac n2)}.
\)
One may consider the effect of additional prior on $D$. For now, we focus on the marginal likelihood itself, and apply negative logarithmic transform to obtain a loss function. This leads to the following optimization problem:
\[\label{eq:dc_program}
 \inf_{D\in \mathcal D} 
& -\frac12\log|D|  +\frac12\log|X^\top X +D|\\
& 
+(a_0+\frac n2)\log  \left\{ b_0+\frac12\Bigl(y^\top y- y^\top X ( X^\top X+D)^{-1}X^\top y\Bigr)\right\}
.
\]
In the above, we further impose a space constraint, $\mathcal D \subseteq \mathbb{D}=\{\diag(d_1,\dots,d_p): d_j > 0, j=1,\dots,p\}$ and $\mathcal D$ is convex.

\begin{figure}[h]
  \centering
  \includegraphics[width=\linewidth]{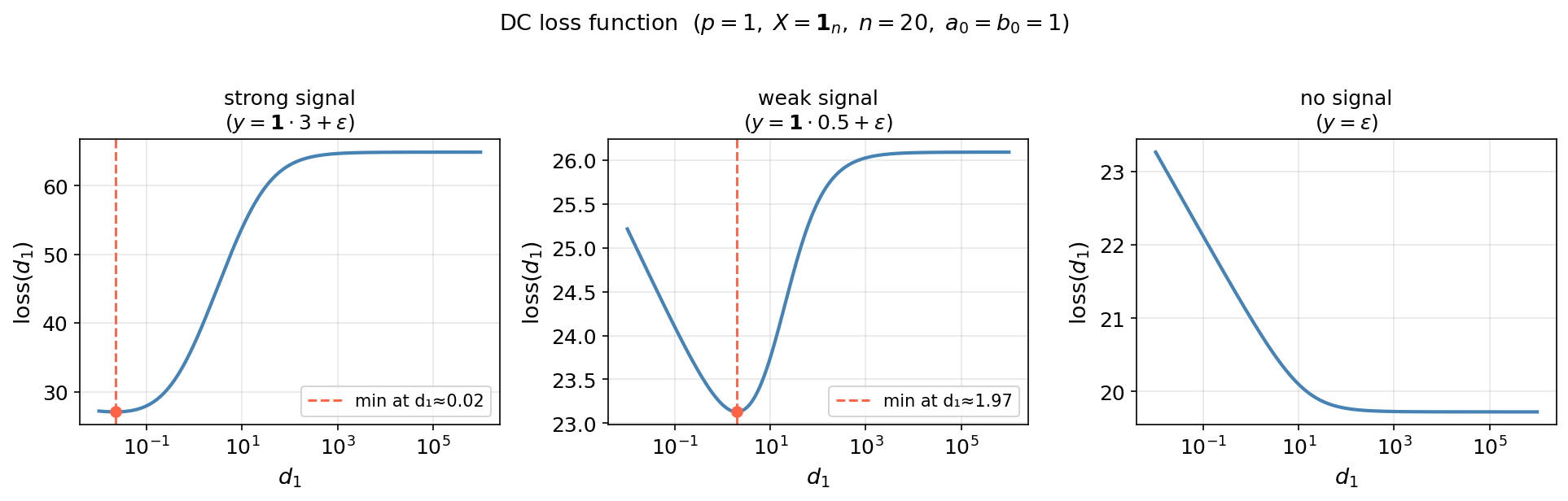}
  \caption{The loss function \eqref{eq:dc_program} for $p=1$, $X=\mathbf{1}_n$, $n=20$, $a_0=b_0=1$,
           under three signal strengths. Red dashed lines mark the interior minimizer. Note under the no-signal case, the optimum is at $d_1\to \infty$, corresponding to shrinking $\theta_1$ to zero.}
  \label{fig:loss1d}
\end{figure}
Maximizing the marginal likelihood over $D$ can be viewed as a type-II maximum likelihood approach \citep{pena2020restricted}.
The marginal likelihood maximized over $D$ has an automatic variable selection mechanism. For illustration, consider a simple setting $p=1$ with $X=\mathbf{1}_n$, the $n$-vector of ones, and $D=\diag(d_1)$. The loss simplifies to
\( -\tfrac{1}{2}\log d_1
  + \tfrac{1}{2}\log(n+d_1)
  + \Bigl(a_0+\tfrac{n}{2}\Bigr)
    \log\left\{
      b_0 + \tfrac{1}{2}\!\left(y^\top y - \frac{(\mathbf{1}^\top y)^2}{n+d_1}\right)
    \right\}, \quad d_1>0.
\)
Figure~\ref{fig:loss1d} plots the loss function for $n=20$, $a_0=b_0=1$, using $y_i=\theta_0+\epsilon_i$ and $\epsilon_i \sim \text{N}(0,1)$ for $i=1,\dots,n$, under three data-generating scenarios: $\theta_0\in\{3,0.5,0\}$ corresponding to strong, moderate, and no signal, respectively.

When the signal is strong or moderate (left and center panels), the loss has a finite interior minimizer $d_1^\star$. The no-signal case (right panel) is qualitatively different: the loss decreases monotonically for all $d_1>0$, with no finite minimizer. This can be seen analytically from the derivative
\(
\frac{\textup{d}\  \text{loss}(d_1)}{\textup{d}\ d_1}
= {-\frac{n}{2\,d_1(n+d_1)}}
  \;+\;
  {\Bigl(a_0+\tfrac{n}{2}\Bigr)
  \frac{ (\mathbf{1}^\top y)^2}{2(n+d_1)^2\!\left[b_0+\tfrac{1}{2}\bigl(y^\top y-\tfrac{(\mathbf{1}^\top y)^2}{n+d_1}\bigr)\right]}},
\)
which is always negative for all $d_1>0$ provided  $(\mathbf{1}^\top y)^2$ is close enough to zero.

Such a diverging behavior of $d_j$ is desirable for variable selection: the prior variance of $\theta_j$ is $(1/d_j)\downarrow 0$, hence the posterior $\theta_j\to 0$ in probability. Similar phenomenon was discovered by \cite{tipping2001sparse} on the partially integrated likelihood $\mathcal L(y\mid D,\sigma^2)$.

 The following theorem gives a sufficient condition for the coordinate-wise diverging behavior for $d_j$ (given the other $d_{k},k\neq j$) in the general $p$-dimensional setting.

\begin{theorem}\label{thm:diverging}
    Let $X=[x_j,X_{-j}]$ and $D=\diag(d_j,D_{-j})$, $\text{RSS}(D)=\bigl(y^\top y-y^\top X(X^\top X+D)^{-1}X^\top y\bigr)$,     $
    P_j(D_{-j})
    =
    I_n-X_{-j}(X_{-j}^\top X_{-j}+D_{-j})^{-1}X_{-j}^\top,
    $
    Then
    \[
    \frac{\partial \text{loss}(D)}{\partial d_j}
    =
    -\frac{x_j^\top P_j(D_{-j})x_j}
    {2\,d_j\bigl(d_j+x_j^\top P_j(D_{-j})x_j\bigr)}
    +
    \frac{a_0+n/2}
    {2\left\{
    b_0+\frac12\text{RSS}(D)
    \right\}}
    \cdot
    \frac{\bigl(x_j^\top P_j(D_{-j})y\bigr)^2}
    {\bigl(d_j+x_j^\top P_j(D_{-j})x_j\bigr)^2},
    \]
    is negative over $d_j>0$ when $$
\bigl(x_j^\top P_j(D_{-j})y\bigr)^2
<
\frac{
\left\{
b_0+\frac12\,y^\top P_j(D_{-j})y
\right\}
\,x_j^\top P_j(D_{-j})x_j}{
(a_0+n/2)
}.    $$
    \end{theorem}

    \begin{remark}
    From the above, we can see immediately two important points. First, even though the coordinate-wise minimizer could be obtained at finite value, as shown in Figure~\ref{fig:loss1d}, the loss function is not convex in $D$. Second, the coordinate-wise diverging behavior means that when the true data generating model corresponds to some $\theta_j=0$, and if such a sparsity pattern is recoverable by the optimizer of the loss function, then  the corresponding $d_j$ will diverge to infinity.
    \end{remark}

    In the remaining of the article, we will impose constraints on $\mathcal D$ to be closed and bounded. In canonical Bayesian variable selection, the boundedness is often not needed, as the prior has a probabilistic regularization effect against infinite values; nevertheless, for optimization, the boundedness serves two important roles: (i) it ensures the optimum is always attainable at finite value; (ii) it allows us to impose useful structured sparsity assumptions, to be discussed in a later section. As a result, if we have a monotone loss function along $j$-th coordinate, we will obtain a finite $d_j$ on the boundary of $\mathcal D$.

\subsection{Difference-of-convex algorithm}
We now first rewrite the loss function as the difference of two convex functions over $\mathbb{D}$:
\[
& g(D)=-\frac12\log|D|+ I_{\mathcal D}(D)
=-\frac12\sum_{j=1}^p \log d_j + I_{\mathcal D}(D),\\
& h(D)
=
-\frac12\log| X^\top X+D|
-(a_0+\frac n2)\log  \left\{ b_0+\frac12\Bigl(y^\top y- y^\top X (X^\top X+D)^{-1}X^\top y \Bigr)\right\},
\]
where $I_{\mathcal D}(D)$ is the indicator function of $\mathcal D$, $I_{\mathcal D}(D)=0$ if $D \in \mathcal D$ and $+\infty$ otherwise.
The convexity of $g(D)$ is obvious, and the convexity of $h(D)$ is due to the fact that $-y^\top X (X^\top X+D)^{-1} X^\top y$ is concave, and log is increasing and concave.

Based on the above difference-of-convex structure, we employ the difference of convex functions (DC) algorithm.
Starting from some $D^{(0)} \in \mathcal{D}$, at iteration $t$, DC algorithm replaces $h$ by its first-order affine approximation at $D^{(t)}$, and computes
\[\label{eq:dca_subproblem}
D^{(t+1)} \leftarrow \underset{D \in \mathbb{D}}{\arg\min} \left\{ g(D) -h(D^{(t)})- \langle \nabla h(D^{(t)}), D - D^{(t)} \rangle \right\}
\]
until convergence.
Since $g$ is convex, $h(D^{(t)})$ and its gradient are considered as given in the subproblem; thus, \eqref{eq:dca_subproblem} is convex hence the update can be obtained efficiently. For our specific problem, since $g(D)$ is a simple log function subject to $\mathcal D$ constraint, we have:
\[\label{eq:dca_update_general}
    d_j^{(t+1)}
    \leftarrow
    P_{\mathcal D} \!\left( \diag\{
    -\frac{1}{2[\nabla h(D^{(t)})]_{jj}} \}_{j=1}^p \right),
\]
where $P_{\mathcal D}(x)$ is the projection onto $\mathcal D$, and the gradient of $h(D)$ is given by
$$
   [\nabla h(D^{(t)})]_{jj}
    =
    -\frac12\Big[(X^\top X+D^{(t)})^{-1}\Big]_{jj}
    -\frac{(a_0+n/2) \left(\Big[(X^\top X+D^{(t)})^{-1}X^\top y\Big]_j\right)^2}{2\left\{\,b_0+\frac12\Big(y^\top y-y^\top X (X^\top X+D^{(t)})^{-1}X^\top y\Big)\right\}}
    .
$$
For high-dimensional problems with $p\gg n$, one may apply Woodbury formula to reduce the cost $(X^\top X+D^{(t)})^{-1}=
(D^{(t)})^{-1} - (D^{(t)})^{-1}X^\top \Bigl(I_n+X(D^{(t)})^{-1}X^\top\Bigr)^{-1} X(D^{(t)})^{-1}$. Therefore, the gradient can be computed with $O(\min(n,p)^3)$ complexity.  

To be more concrete, for simple variable selection tasks, we will use simple box constraint $\mathcal D$: $0\le a\le d_j\le b$ for all $j=1,\dots,p$, which leads to the following closed-form update:
\[\label{eq:dca_update}
    d_j^{(t+1)}
    \leftarrow
    P_{[a,b]}\!\left(
    -\frac{1}{2[\nabla h(D^{(t)})]_{jj}}
    \right),
    \qquad j=1,\dots,p,
\]
where 
$P_{[a, b]}(x) = \min\{b, \max\{a, x\}\}$ for $x\in \mathbb{R}$
is the projection onto the interval $[a, b]$.   We see that the algorithm \eqref{eq:dca_update_general} is simple and tuning-free, yet it accounts for the uncertainty of $\theta$ and $\sigma^2$ due to its operation on the marginal likelihood.

In the following two subsections, we show that the above algorithm not only monotonically decreases the loss function $g(D)-h(D)$, hence leading to convergence, but also enjoys global-optimum convergence guarantee under mild conditions.

\subsection{Convergence and acceleration}\label{sec:acceleration}
We first show that the DC algorithm (under a general convex set $\mathcal D$) using update rule \eqref{eq:dca_subproblem} is guaranteed to converge. The argument follows standard optimization theory techniques, but is informative to the readers unfamiliar with the topic.

Since \eqref{eq:dca_subproblem} is convex, by first-order optimality condition, we have
\(
\nabla h(D^{(t)})\in \partial g(D^{(t+1)}),
\)
where $\partial g(x)$ is the subdifferential of $g$ at $x$: the collection of all subgradients $\{v_x \in \mathbb{D}: g(y) \ge g(x)+ \langle v_x, y-x\rangle \}$ for all $y\in \mathbb{D}$. The subdifferential $\partial g(x)$ reduces to a single element for gradient when $g$ is differentiable at $x$.

Using the definition of subdifferential, we have
\(
g(D^{(t)}) \ge g(D^{(t+1)})+ \langle \nabla h(D^{(t)}), D^{(t)}-D^{(t+1)}\rangle.
\)
Since $h$ is convex and differentiable, we have $h(D^{(t+1)})\ge h(D^{(t)})+\langle \nabla h(D^{(t)}),D^{(t+1)}-D^{(t)}\rangle.$ Adding the above two inequalities, we obtain
\(
g(D^{(t)})-h(D^{(t)})\ge g(D^{(t+1)})-h(D^{(t+1)}).
\)
This shows that the DC algorithm is a descent algorithm: the loss function $g(D)-h(D)$ is monotonically non-increasing. As $t\to \infty$, provided the loss function is bounded below (such as when $d_j\ge a>0$ for all $j=1,\dots,p$), the loss function values converge.

Next, we know the algorithm converges at the fixed point of \eqref{eq:dca_subproblem}, that is when $D^{(t)}=\underset{D \in \mathbb{D}}{\arg\min} \left\{ g(D) - \langle \nabla h(D^{(t)}), D  \rangle \right\}$ for $t$ large enough. Denoting such $D^{(t)}$ as $D^\dagger$, and applying the differential rule of indicator function, we know the fixed point corresponds to
\(
0\in \nabla \tilde g(D^\dagger)+N_{\mathcal D}(D^\dagger)-\nabla h(D^{(t)}),
\)
where $\tilde g(D) = -\frac{1}{2}\sum_{j=1}^p \log d_j$, and $N_{\mathcal D}(D^\dagger)
:=
\Bigl\{
U:\ \langle U,\; D-D^\dagger\rangle \le 0
\quad \forall D\in\mathcal D
\Bigr\}$ is the normal cone.  When  $D^\dagger$ is in the interior of $\mathcal D$, $N_{\mathcal D}(D^\dagger)=\{0\}$, the above condition reduces to $\nabla \tilde g(D^\dagger)= \nabla h(D^\dagger)$; nonetheless, the general normal cone representation allows $D^\dagger$ to be on the boundary of $\mathcal D$.

For tracking and diagnosing the convergence of $D$, we can use the residual distance:
\(
    \lambda(D)=\text{dist}[\nabla h(D) - \nabla \tilde g(D), N_{\mathcal D}(D)]
\)
and see how far is from zero. The notation $\text{dist}(x,C)$ here means the shortest Euclidean distance between $x$ and its closest point in $C$.

For the box constraint, the normal cone is given by
$
    N_{\mathcal D}(D)
    =
    \Bigl\{
    U=\diag(u_1,\dots,u_p):
    \ u_j=0 \text{ if } a<d_j<b,\;
    u_j\le 0 \text{ if } d_j=a,\;
    u_j\ge 0 \text{ if } d_j=b
    \Bigr\}.    
$
Therefore, we have the residual distance squared as:
\(
\lambda^2(D)= \sum_{j=1}^p \lambda_j^2(D), \qquad
\lambda_j(D)=
\begin{cases}
s(D), & a<d_j<b,\\[0.4em]
\max\{ [\nabla h(D) - \nabla \tilde g(D)]_{jj},\,0\}, & d_j=a,\\[0.4em]
\max\{-[\nabla h(D) - \nabla \tilde g(D)]_{jj},\,0\}, & d_j=b.
\end{cases}
\)
Therefore, a potential acceleration to the convergence is boundary exploration: at iteration $t$, if $d^{(t)}_j\le \tau_1$, then we may propose $d^{(prop)}_j=a$; similarly, if $d^{(t)}_j\ge \tau_2$, then we may propose $d^{(prop)}_j=b$. Then we evaluate if $g(D^{(prop)})-h(D^{(prop)})<g(D^{(t)})-h(D^{(t)})$. If so, we accept the proposal and set $D^{(t+1)}=D^{(prop)}$.

In practice, since we use a low $a=10^{-4}$, we do not observe any optimal $d_j$ converging to $a$, but we do frequently see convergence to $b=10^4$. Therefore, we set $\tau_1 = a$ but $\tau_2 = 0.1 b$. To prevent sudden large changes in $D$ that causes increase of loss function, we allow at most 20 coordinates $d_j$ to be updated to the boundary in any given iteration.

\subsection{Global optimality}

Next, it is natural to question whether the converged value is the global minimum. We now establish a set of sufficient conditions.

For ease of notation, we denote the loss function $f(D) = \tilde{g}(D) - h(D)$
with $\tilde{g}(D) = -\frac{1}{2}\sum_{j=1}^p \log d_j$
, the set of global minimizers $
\mathcal S:=\arg\min_{D\in\mathcal D} f(D)$, any global optimizer $D^\star$, and the minimum value by $
f^*:=\min_{D\in\mathcal D} f(D)$. Since $f$ is differentiable, the optimal $D^\star$ has the following necessary  condition:
\(
\text{(Stationarity) }\;\exists U^\star \in N_{\mathcal D}(D^\star) \text{ such that } \nabla f(D^\star)+U^\star=0.
\)
We refer to $U\in N_{\mathcal D}(D): U= -\nabla f(D)$ as the {\em normal-cone multiplier} for $D$, based on the origin of $U$ from Karush–Kuhn–Tucker conditions. Despite the necessity, the above stationarity condition is not sufficient: the existence of normal-cone multiplier $U$ for $D$ cannot guarantee the optimality of $D$, since $f$ is not convex.

On the other hand, for those non-optimal points satisfying stationarity $\bar{\mathcal S} := \{D \in \mathcal D : U \in N_{\mathcal D}(D),\; \nabla f(D) + U = 0,\; f(D) > f^\star\}$. If these non-optimal points in $\bar{\mathcal S}$ all have $\|\nabla f(D)\|_F > r$ by some $r>0$, then the DC algorithm can be robust to such points and converges to the global minimum.

\begin{theorem}\label{thm:global_nonunique}
    Let $\mathcal D$ be compact, closed, and convex subset of     $\mathbb{D} = \{\diag(d_1,\dots,d_p) : d_j > a\}$, with $a>0$. 
    Let $r>0$ be a constant such that
    \(
    \|\nabla \tilde g(D^{(t+1)})-\nabla h(D^{(t)}) \|_F\le r
    \qquad\text{for all }t \ge t_0,
    \)

    Suppose for $\mathcal D$ and optimal set $\mathcal S$, there exist constants $M>0$, $m>0$, $\rho>0$:
    \begin{enumerate}[label=(A\arabic*)]
        \item For every $D\in \mathcal D:\dist(D,\mathcal S)\le \rho$,
        $f(D)-f^* \le M\,\dist(D,\mathcal S)^2.$
        \item For every $D\in\mathcal D:\dist(D,\mathcal S)\le \rho$, and for every
        $U\in N_{\mathcal D}(D):\|U\|_F\le r$,
        $
        \|\nabla f(D) +U\|_F
        \ge
        m\,\dist(D,\mathcal S).
        $
        \item   
        If $D\in\mathcal D$ and $U\in N_{\mathcal D}(D):\|U\|_F\le r$ satisfy
        $
        \nabla f(D) +U=0,
        $
        then $D\in\mathcal S$.
    \end{enumerate}
    Then there exists a constant $\mu>0$ such that the DC algorithm iterates \eqref{eq:dca_subproblem} satisfy
    \(
    f(D^{(t+1)})-f^\star
    \le
    \frac{1}{1+\mu}\bigl(f(D^{(t)})-f^\star\bigr),
    \) 
    for all $t \ge t_0$.
  \end{theorem}
    \begin{remark}   
      We provide the details of $\mu$ in the proof of the theorem.
        The conditions (A1) and (A2) describe the local geometry of the problem around
        the solution set $\mathcal S=\arg\min_{D\in\mathcal D} f(D)$.
        In particular, (A1) is a local quadratic growth condition, stating that the
         gap $f(D)-f^\star$ is controlled by the squared distance from $D$
        to $\mathcal S$, while (A2) is a residual error-bound condition, asserting
        that the residual
        $
        \nabla f(D) +U
        $
        grows at least linearly with $\dist(D,\mathcal S)$.
        Condition (A3) states that all non-optimal stationary points in
        $\mathcal D$ have normal-cone multiplier norm above $r$, as we discussed before the theorem. To be rigorous, we do not claim a uniform convergence rate since $r$ is dependent on $t_0$: $t_0$ can be interpreted as the number of warm-up iterations for getting $D^{(t)}$ into a reasonable set so that $\|\nabla \tilde g(D^{(t+1)})-\nabla h(D^{(t)}) \|_F$ is not too large going forward.
    \end{remark}

    Examining the conditions (A3), we can see that the theorem effectively rules out those pathological problems with non-optimal $D$ in the $\mathcal D$ interior and $\nabla f(D)=0$. This exclusion is not surprising for the DC algorithm, as we can see from the update rule \eqref{eq:dca_subproblem} --- the algorithm would converge at $D^{(t)}$ if $\nabla \tilde g(D^{(t)})- \nabla h(D^{(t)})=0$, hence could be stuck at such a non-optimal point if it exists.
    
    Fortunately, by Theorem \ref{thm:diverging} and the remark that follows, we know if a non-optimal stationary points correspond to an alternative choice of sparsity patterns (with at least one $\theta_j\approx 0$), then the corresponding $D$ will be at the boundary of $\mathcal D$, hence $\|\nabla f(D)\|_F>0$ and can be excluded from the DC algorithm.

    The above theorem and its proof can be readily extended to the setting where one optimizes the marginal posterior of $D$ instead of the marginal likelihood, provided that the added prior term $-\sum_{j=1}^p \log \pi_0(d_j)$ can be expressed as a difference of convex functions, and assumptions (A1), (A2), and (A3) are satisfied. 
    
    On the other hand, it is important to note that the posterior mode and mean can exhibit markedly different behaviors in terms of variable selection. Consider, for example, the horseshoe prior, where $\sqrt{1/d_j} \sim \text{Cauchy}_+(0, \tau)$ \citep{carvalho2010}. Accounting for the Jacobian, the negative log-prior takes the form
    $
        -\log \pi_0(d_j)
        =
        \frac{1}{2} \log d_j
        + \log(1+\tau^2 d_j)
        + \text{const},
    $
    where the first term, when summed over $j$, completely cancels $\tilde g(D)$, resulting in a concave loss function over the entire domain $\mathbb{D}$. As a consequence, when a box constraint is imposed on $D$, each optimal $d_j^{\star}$ is forced to either the lower or upper boundary ($a$ or $b$). Empirically, we found that this binary, coordinate-wise extremal form of the posterior mode performs poorly for signal recovery. In contrast, the posterior mean under the horseshoe prior is well known for its strong ability to discriminate between signal and noise.

    \subsection{Inducing structured sparsity via additional constraints}\label{sec:structured}

    We now discuss useful extensions for inducing \emph{structured sparsity}. Structured sparsity means that the zero pattern is not arbitrary, but is instead constrained to respect known structure in the features, such as groups, overlaps, or an ordering across indices. In our framework, this is encoded by requiring the vector of sparsity parameters $(d_1,\dots,d_p)$ to belong to a closed convex set $\mathcal C\cap [a,b]^p$, so that those $j:d_j \ge \tau_2$ (with $\tau_2$ large enough to correspond to near-zero $\theta_j$) adhere to structured sparsity patterns.
    
    Different structural assumptions lead to different choices of $\mathcal C$ and, consequently, different projection algorithms. Regardless, if the projection lacks a closed-form solution, one could use the alternating projection algorithm. For $\mathcal C\cap [a,b]^p=\cap_{k=1}^{K} \mathcal C_k \cap [a,b]^p$, one could start from a value $\tilde d: \tilde d_j= -([2\nabla h(D^{(t)})]_{jj})^{-1}$ for each $j$, and then iterate the following update until convergence of $\tilde d$:
    \(
    \tilde d \leftarrow P_{[a,b]^p}\circ P_{\mathcal C_1}\circ \cdots \circ P_{\mathcal C_K} (\tilde d), \quad k=1,\dots,K.
    \)
    The alternating projection algorithm converges rapidly since the target set is an intersection of convex sets.
    
    We now provide a few useful examples of structured sparsity constraints. First, for a non-overlapping group structure, let $\{G_1,\dots,G_M\}$ be a partition of $\{1,\dots,p\}$. A natural constraint is to require all coordinates within the same group to be simultaneously close to zero, if the corresponding features need to be excluded at the same time. This can be encoded as the following constraint:
    \(
    \mathcal C =\{ d: d_i=d_j,\quad \forall i,j\in G_m,\quad m=1,\dots,M \}.
    \)
   The set is a linear subspace, and the Euclidean projection is obtained by averaging within each group:
    \(
    (P_{\mathcal C}(\tilde d))_i
    =
    \frac{1}{|G_m|}\sum_{k\in G_m} \tilde d_k,
    \qquad i\in G_m.
    \)
    
    For overlapping groups, let $\mathcal G=\{G_1,\dots,G_M\}$ be a collection of possibly overlapping subsets of $\{1,\dots,p\}$. 
    For overlapping groups, we first decompose the index set into the smallest disjoint regions $(\mathcal R_1,\dots,\mathcal R_L)$ induced by the overlaps of $G_1,\dots,G_M$. Each region consists of coordinates that belong to exactly the same combination of groups, and we impose that $d_i$ is constant within each region. We then require regions in deeper intersections to have equal or stronger shrinkage than regions in shallower ones. In this way, coordinates that lie in multiple overlapping groups are shrunk more strongly than those that lie in only one group. The projection is obtained by first averaging within each region and then solving the resulting isotonic regression problem over these regions. The projection is therefore given by
\(
& P_{\mathcal C}(\tilde d)
=
\arg\min_z \frac12\sum_{i=1}^p (z_i-\tilde d_i)^2
\\ 
& \text{subject to}
\quad
\begin{cases}
d_i=d_j, & \text{if } i,j \text{ belong to the same region $R_k$},\\
d_j\le d_i, & \text{if } j \text{ lies in a strictly deeper overlap region than } i.
\end{cases}
\)
The solution can be efficiently solved by isotonic regression algorithms \citep{dykstra1981,best1990}.

Lastly, we discuss the case where there is a monotonicity of shrinkage strength over multi-index. Suppose the parameters are indexed by $u\in\mathcal I$, where $u$ may be a two-dimensional or higher-dimensional index. We can impose a partial-order monotonicity constraint of the form
    $
    d_u \le d_v
     \text{if } u \preceq v,
    $
    where $\preceq$ is a prescribed partial order reflecting the desired structure. Examples include coordinatewise monotonicity, or radial monotonicity in which $d_u$ is required to decrease as $u$ moves farther away from some reference origin. The projection is then
    \(
    P_{\mathcal C}(\tilde d)
    =
    \arg\min_{z}
    \frac{1}{2}\sum_{u\in\mathcal I}(z_u-\tilde d_u)^2
    \quad
    \text{subject to }
    z_u\le z_v, \text{ if } u\preceq v.
    \)
    This is a generalized isotonic regression problem on a partially ordered set.
    
    Since the last projection will be used in our spatial data application, we now discuss the case where the $u\preceq v$ order is induced by a scalar score $s(u)$ (for example, distance to the origin).
    The projection reduces to weighted isotonic regression after averaging over ties in $s(u)$. Specifically, if
    $
    S_t=\{u\in\mathcal I : s(u)=t\}
    $
    denotes the level set associated with score $t$, and
    $
    \bar d_t={|S_t|^{-1}}\sum_{u\in S_t} \tilde d_u,
    $
    then the projection is obtained by solving
    $
    \hat g
    =
    \arg\min_{g_1\ge g_2\ge\cdots}
   \sum_t |S_t|(g_t-\bar d_t)^2,
    $
    and updating $\tilde d_u \leftarrow \hat g_t$ for all $u\in S_t$. In this case, one can apply the weighted pool-adjacent-violators algorithm (PAVA) to solve the projection, and lastly, project to $[a,b]^p$ in the last step. The projection algorithm is non-iterative, as it takes only one pass of averaging, PAVA, and box-truncation, hence is efficient.

\section{Numerical experiments}
In this section, we conduct numerical experiments using both synthetic and real data to illustrate the performance of the proposed DC algorithm.

\subsection{Simulation}

We compare the proposed difference-of-convex (DC) algorithm against two alternative optimization algorithms: the EM algorithm for automatic relevance determination (ARD) \cite{tipping2001sparse}, and projected gradient descent (PGD) for the loss function \eqref{eq:dc_program}.
All three algorithms use the same box constraint $\mathcal{D} = \{D : 10^{-4} \le d_j \le 10^4,\,j=1,\dots,p\}$ with hyperparameters $a_0=b_0=1$, and are initialized at $D^{(0)}=I_p$.
Each algorithm is considered converged when the relative change in the loss function (expected loss function in EM algorithm) satisfies $|f(D^{(t+1)})-f(D^{(t)})|/|f(D^{(t)})| < 10^{-5}$ for two consecutive iterations.

\begin{figure}[H]
  \centering
  \includegraphics[width=\linewidth]{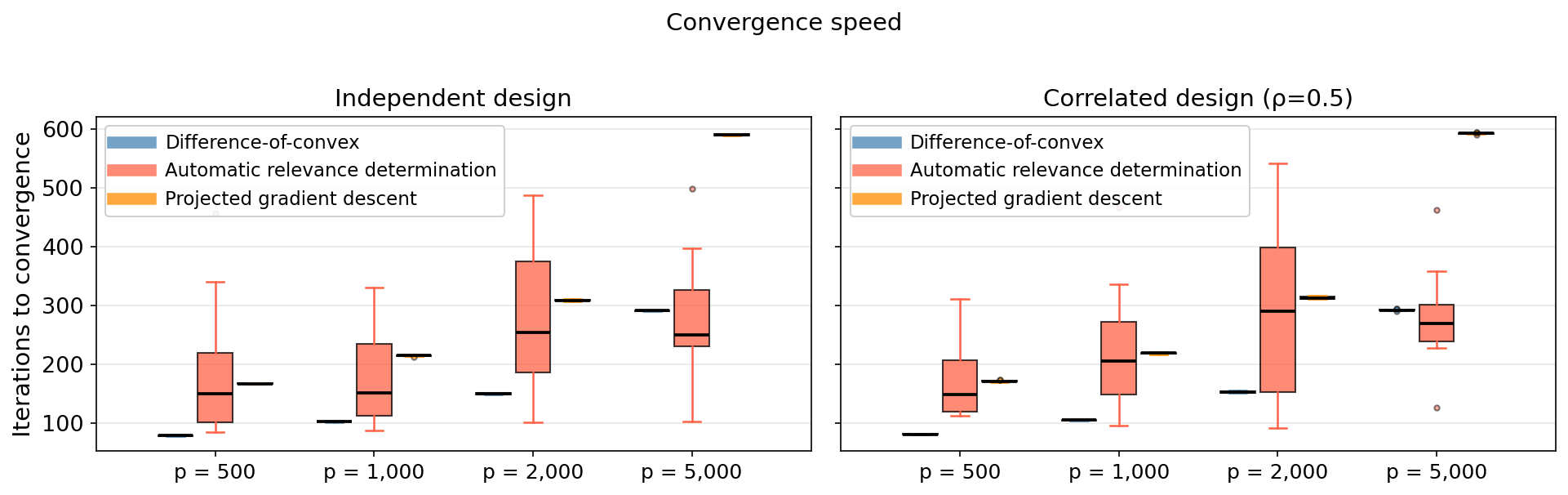}
  \caption{Iterations to convergence for the difference-of-convex algorithm (DC), automatic relevance determination (ARD), and projected gradient descent (PGD) across four dimensionality settings under independent and correlated designs ($\rho=0.5$). Each boxplot summarizes 10 independent replications. The numbers for both difference-of-convex and projected gradient descent have low variance, hence the boxes are very narrow.}
  \label{fig:iters}
\end{figure}

We consider four dimensionality settings, $p\in\{500,1000,2000,5000\}$, each paired with two design structures: an {independent} design ($X_{ij}\overset{\mathrm{iid}}{\sim}\mathrm{N}(0,1)$) and a {correlated} design generated from a Toeplitz covariance with $\mathrm{Cov}(X_{ij},X_{ik})=0.5^{|j-k|}$.
All columns of $X$ are standardized to unit variance.
In every scenario the sample size is $n=200$, the number of true signals is $s=5$, their nonzero coefficients are drawn uniformly from $\{-3,-2.5,2.5,3\}$, and the noise level is $\sigma^2=1$.
Each scenario is replicated 10 times with independent random seeds.
Variable selection is based on the posterior mean estimator $\hat\theta=(X^\top X+D)^{-1}X^\top y$ evaluated at the converged $\hat D$: predictor $j$ is selected if $|\hat\theta_j|>0.1$.

All three algorithms achieve perfect variable selection (true positive rate $= 1$, false positive rate $= 0$) across every scenario and replication, confirming the theoretical global-optimality guarantees for DC.

\begin{figure}[H]
  \centering
  \includegraphics[width=\linewidth]{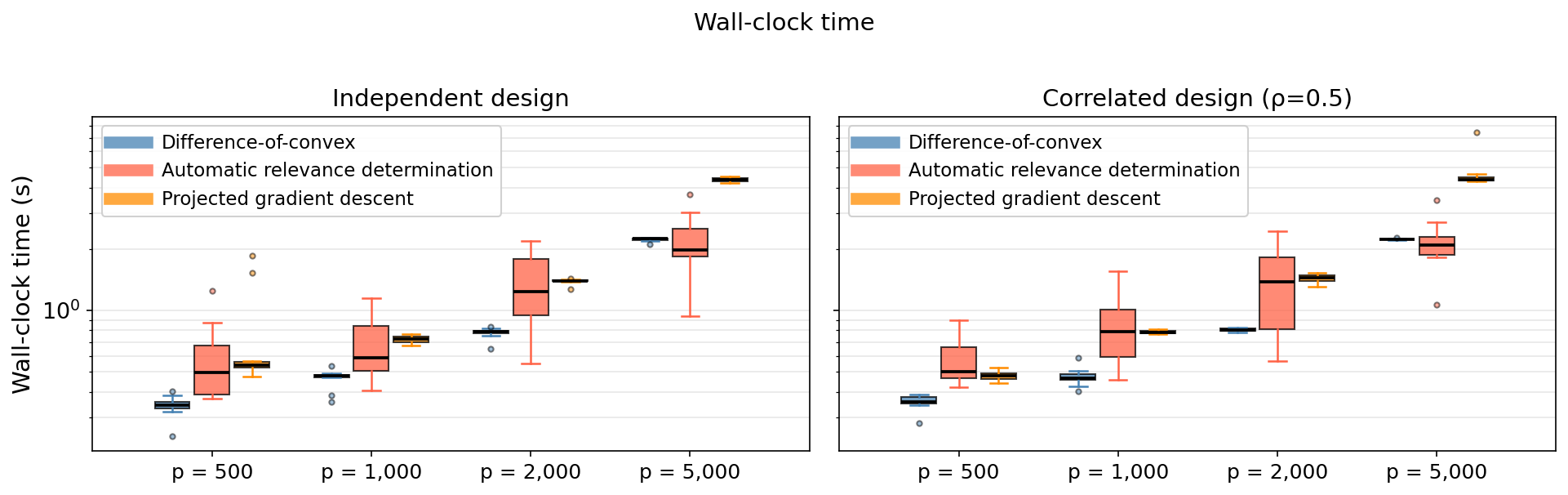}
  \caption{Wall-clock time (seconds, log scale) for the same three algorithms and scenarios as Figure~\ref{fig:iters}.}
  \label{fig:walltime}
\end{figure}

The algorithms differ markedly in convergence speed, as summarized in Figures~\ref{fig:iters} and \ref{fig:walltime}.
Difference-of-convex requires the fewest iterations: roughly 79--103 for $p\le 1000$ and up to 292 for $p=5000$, with negligible across-replication variance.
Automatic relevance determination converges in a comparable number of iterations on average but exhibits substantially higher variability (standard deviation up to $\pm140$ iterations), reflecting the sensitivity of the coordinate-wise EM step size to the local curvature.
Projected gradient descent with the log-space scaled step size converges in 167--593 iterations, monotonically increasing with $p$, and with very low replication-to-replication variance. On average, the projected gradient descent takes more iterations than the other two algorithms, due to the non-adaptiveness of the chosen step size. We also experiment with adaptive step size for projected gradient descent with backtracking line search, the algorithm takes considerably more time in each iteration, hence we do not include the results here. In terms of wall-clock time the difference-of-convex algorithm is the fastest, completing all scenarios in under 2.3 seconds per replicate on a 6-core CPU of a Macbook Pro. Automatic relevance determination is slightly slower but comparable, while projected gradient descent takes the longest time.

\subsection{Scalability to high-dimensional data}
We apply the DC algorithm to the Dorothea dataset \citep{guyon2004result}, a well-known benchmark dataset for variable selection due to its high dimensionality.
The data represent a drug-discovery task: each sample is a chemical compound and a
response score $y_i$ indicates the effectiveness of the compound against a
target protein. We treat $y_i$ as continuous in this application.
Each compound is described by $100{,}000$ binary structural features;
after removing the 8{,}402 features that are identically zero across the combined
training and validation sets, $p = 91{,}598$ features remain.
We combine the training and validation sets, giving $n=1{,}150$ observations.
Following standard practice for the linear model, both $X$ and $y$ are
standardized to zero mean and unit variance column-wise.

We run the DC algorithm with box constraint $\mathcal D =
\{D : 10^{-4}\le d_j \le 10^4,\, j=1,\dots,p\}$,
hyperparameters $a_0=b_0=1$, and warm-start $d_j^{(0)} = P_{[a,b]}(1/\hat\theta_{j,\text{MLS}}^{2})$
from marginal least-squares estimates $\hat\theta_{j,\text{MLS}}= X_j^\top y/n$.
Since the matrix $I_n + X D^{-1}X^\top$ is high-dimensional, and the inversion cost $O(n^3)$ is still not trivial even with the Woodbury formula, we use GPU acceleration and conjugate gradient approach to solve the linear system at each step. To accelerate convergence toward the boundary, we augment the DC iterations with the acceleration step described in Section \ref{sec:acceleration} at every 100-th iteration.

Figure~\ref{fig:dorothea} summarizes the results after 5{,}000 iterations
(approximately 9 minutes on GPU).
The panel(b) shows the loss function decreasing monotonically and rapidly,
consistent with the descent guarantee of the DC algorithm. At convergence,
$91{,}568$ out of $91{,}598$ features ($99.97\%$) are driven to the upper boundary
$b = 10^4$, with the posterior mean $\hat\theta_j\approx 0$. 
This results in 30 features with $d_j < b$. Figure~\ref{fig:dorothea}(a) displays the estimated $d_j$ values for these features. Additionally, 4 features have $d_j$ values quite close to, but not exactly at, the upper bound $b$.

\begin{figure}[H]
  \centering
  \begin{subfigure}[b]{1\linewidth}
    \centering
    \includegraphics[width=\linewidth]{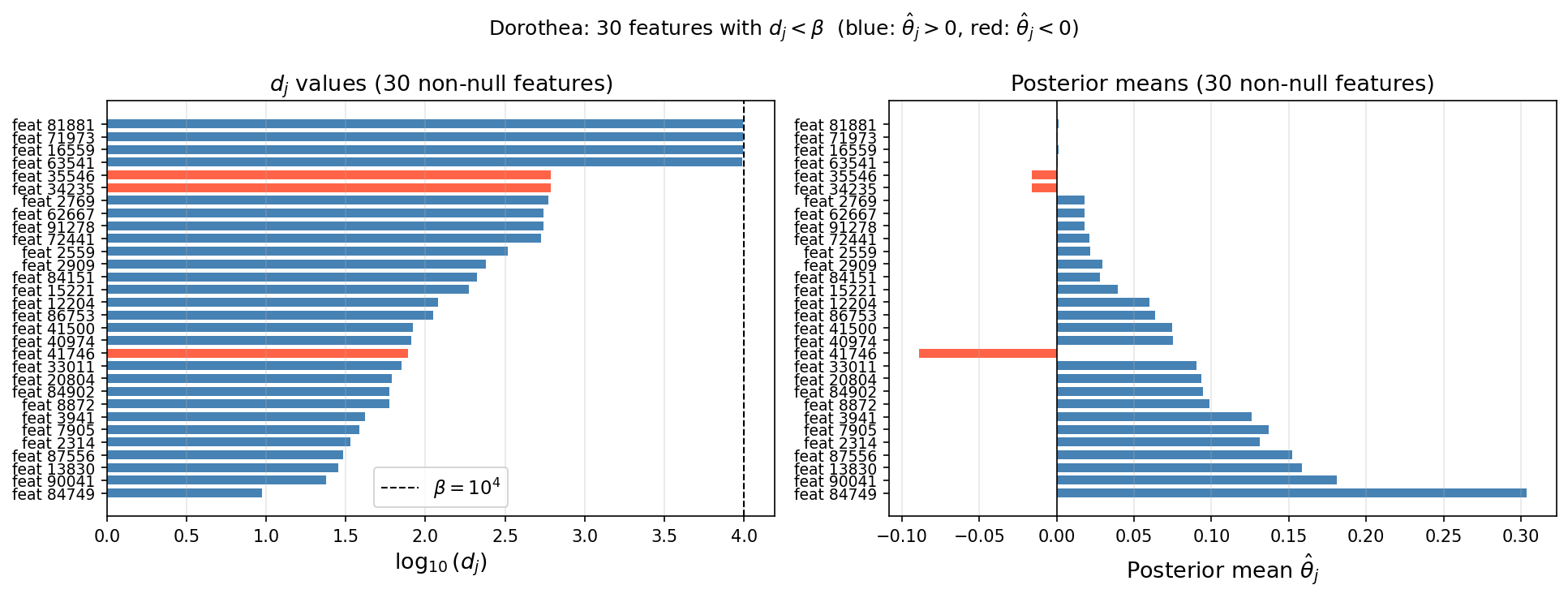}
    \caption{Estimated $d_j$ values for the 30 features with $d_j < b$.}
  \end{subfigure}\\
  \centering
  \begin{subfigure}[t]{0.45\linewidth}
    \centering
    \includegraphics[width=\linewidth,height=4cm]{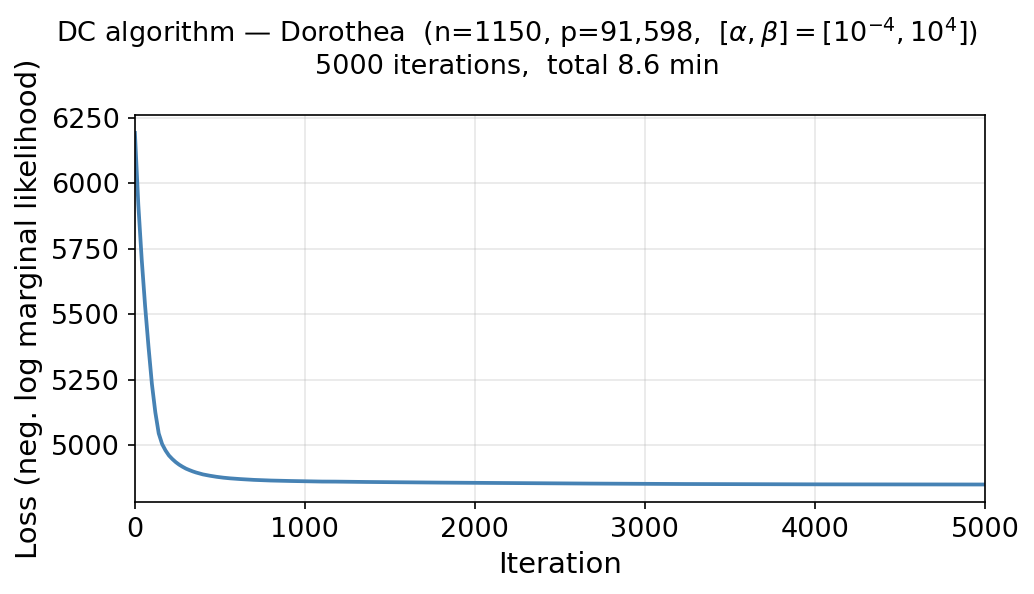}
    \caption{Loss (negative log marginal likelihood) versus DC iteration.}
  \end{subfigure}%
  \hfill
  \begin{subfigure}[t]{0.45\linewidth}
    \centering
    \includegraphics[width=\linewidth,height=4cm]{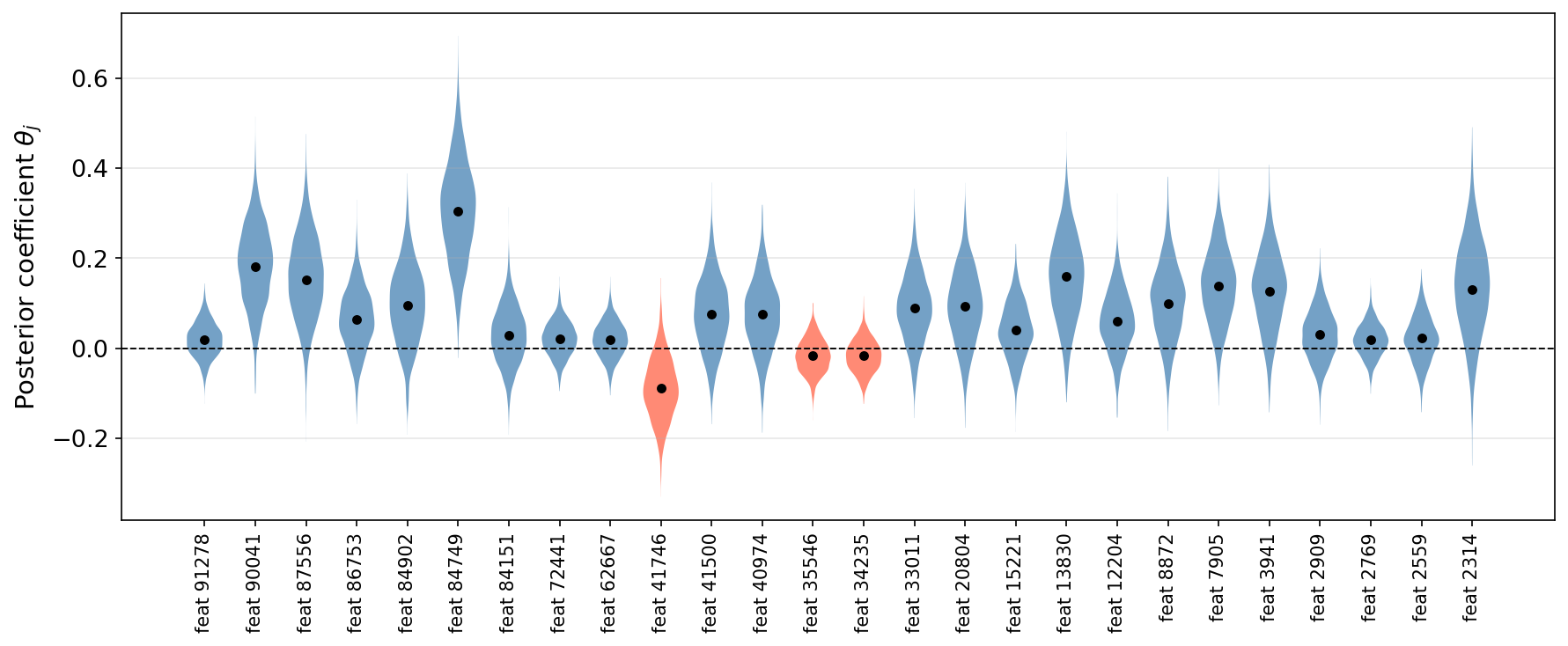}
    \caption{Posterior distributions of the 26 selected features (blue: positive mean, red: negative mean).}
  \end{subfigure}
  \caption{DC algorithm applied to the Dorothea dataset ($n=1{,}150$,
           $p=91{,}598$).}
  \label{fig:dorothea}
\end{figure}

The remaining $26$ features have $d_j \ll b$, corresponding to compounds with non-negligible association with the response.
The panel(c) shows the marginal posterior distributions of the selected coefficients. For comparison with earlier results \citep{lim2015learning}, the lasso selected 9 features, and the group lasso selected 93 features; thus, the DC algorithm yields a number of selected features comparable to these two methods, yet it incorporates the uncertainty of both the regression coefficients and the noise variance.

\section{Spatial data application for quantifying the aftershock risk after earthquakes}

We now apply the DC algorithm in fitting a structured shrinkage model for quantifying the aftershock activities following the 2019 Ridgecrest earthquake sequence in California (mainshock M7.1, epicenter at $35.77^\circ$N, $117.59^\circ$W). Aftershock catalogs are downloaded from the United States Geological Survey (USGS) ComCat API, retaining all events with magnitude $\geq 1.0$ within a $150\,\text{km}$ radius and $30$ days after the mainshock  on July 4-5, 2019. The study region is partitioned into a regular $5\,\text{km}\times 5\,\text{km}$ grid, yielding $n = 2,860$ cells. Each cell $i$ is characterized by its centroid coordinates  and its  distance to the fault plane, roughly matching a $50\,\text{km}$ surface rupture. The response $y_i$ is a log-one-plus transform of some weighted count of aftershocks assigned to cell $i$, which we model as
\(
    & y_i = \theta_i + \varepsilon_i, \quad \varepsilon_i \overset{\mathrm{iid}}{\sim} N(0,\sigma^2), \qquad \theta_i \sim \text{N}(0, d_i^{-1}),\\
    & d_j \geq d_i \quad \text{if } s_j > s_i.
\)
The model is not independent over spatial location $i$, due to the constraint that leads to stronger shrinkage of $\theta_i$ towards zero when $j$ is farther from fault plane than $i$, reflecting the physical expectation the {\em sparsity} of aftershock activities increases with distance. Note that the expectation on sparsity is a distributional assumption about the chance of having mean zero aftershocks, and is different from a simple monotonicity constraint ($\theta_j \leq \theta_i$ if $s_j \leq s_i$, which is arguably too restrictive).

Since the above constraint is an example of the radial monotonicity  discussed in Section~\ref{sec:structured}, with $s(i)$ as the scalar score, the projection algorithm can be applied in each DC step. The DC algorithm converges rapidly in 35 iterations, within 0.20 seconds on a 6-core CPU of a Macbook Pro.

\begin{figure}[H]
  \centering
  \begin{subfigure}[t]{0.48\linewidth}
    \centering
    \includegraphics[width=\linewidth]{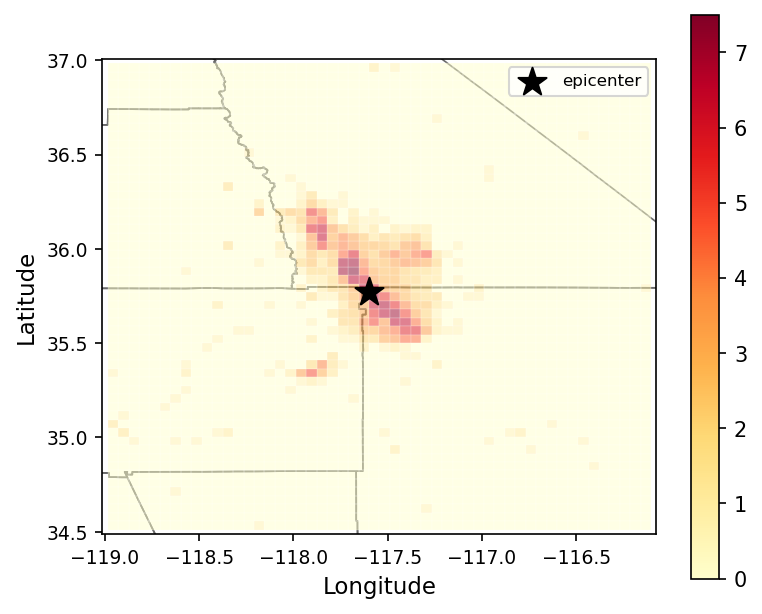}
    \caption{Raw observed aftershock activity scores.}
  \end{subfigure}
  \hfill
  \begin{subfigure}[t]{0.48\linewidth}
    \centering
    \includegraphics[width=\linewidth]{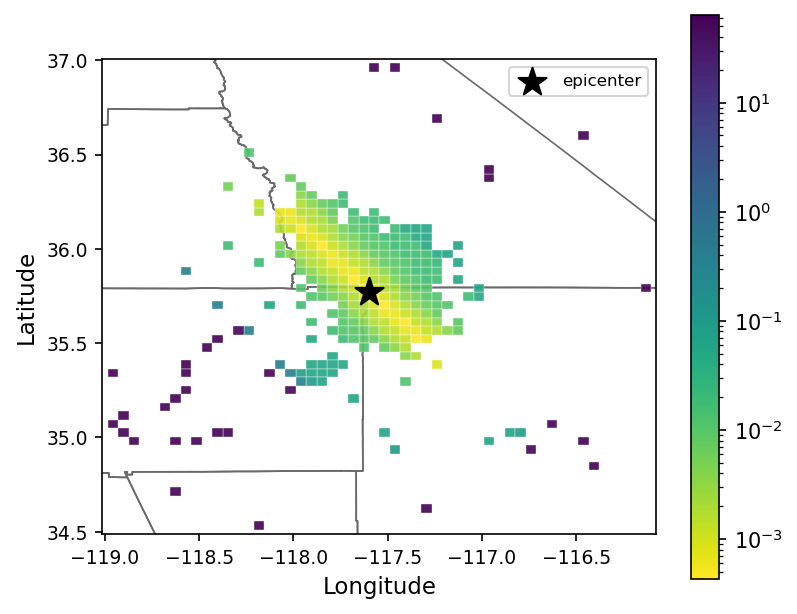}
    \caption{Estimated precision $d_i^\star$ over non-zero cells.}
  \end{subfigure}\\
  \begin{subfigure}[t]{0.48\linewidth}
    \centering
    \includegraphics[width=\linewidth]{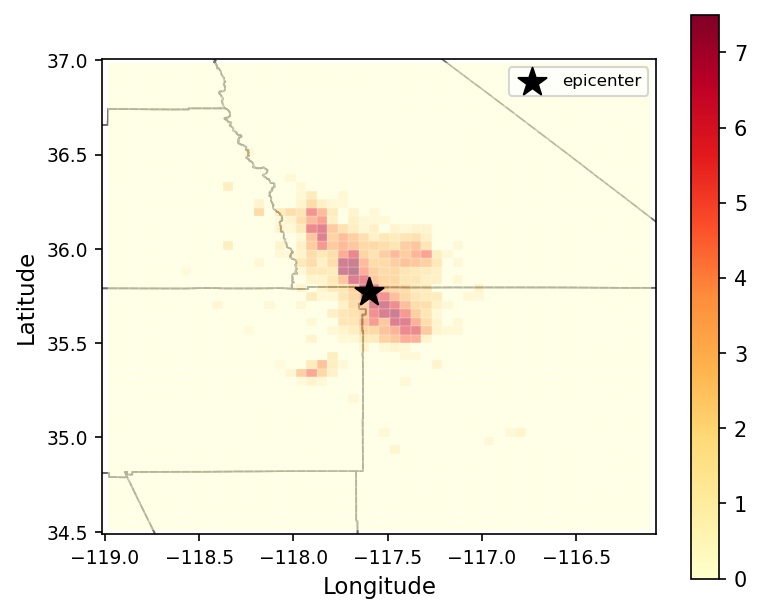}
    \caption{Posterior mean $\hat\theta_i$ with radial monotone shrinkage.}
  \end{subfigure}\hfill
  \begin{subfigure}[t]{0.48\linewidth}
    \centering
    \includegraphics[width=\linewidth]{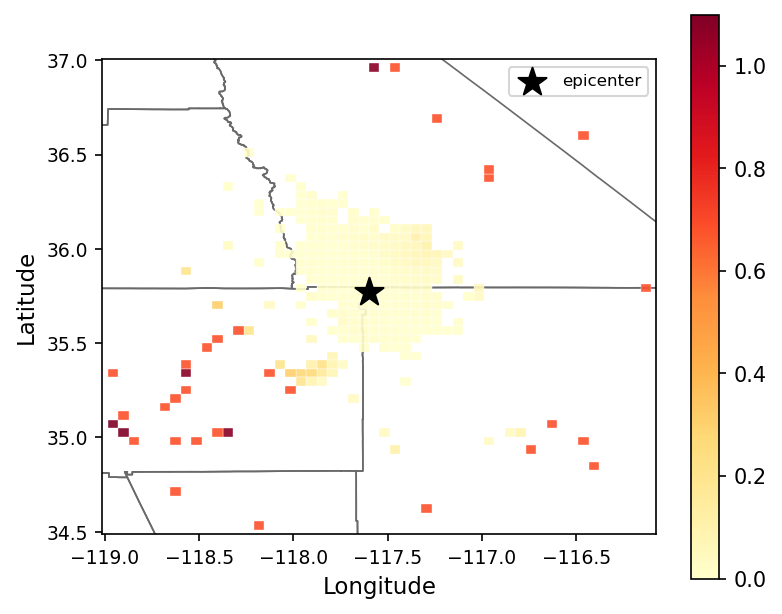}
    \caption{Absolute residuals $|y_i - \hat\theta_i|$.}
  \end{subfigure}
  \caption{DC spatial smoothing of the 2019 Ridgecrest aftershock sequence
    (M7.1, $n=2{,}860$ grid cells, $5\,\text{km}$ resolution, 30-day window).
    California county boundaries are shown in gray. The star marks the epicenter.}
  \label{fig:spatial}
\end{figure}

The raw surface $y_i$ (Figure~\ref{fig:spatial}(a)) reveals that activity is concentrated along the NW--SE fault rupture, with additional scattered aftershocks appearing at the periphery. The estimated $d_i^\star$ values over non-zero $y_i$ cells (Figure~\ref{fig:spatial}(b)) exhibit increasing shrinkage as the distance from the fault plane grows. Among all the cells, $138$ have $d_i^\star = b = 10^4$, corresponding to the strongest shrinkage on $\theta_i$; while most of the remaining $d_i^\star$ are less than $b$, there are $2,606$ cells with posterior mean $|\hat\theta_i| < 0.05$. The median $d_i^\star$ is $64.1$, reflecting an overall moderate degree of shrinkage for those cells with nonzero $\theta_i$. The DC posterior mean $\hat\theta_i$ (Figure~\ref{fig:spatial}(c)) is close to $y_i$ near the fault rupture, but decreases to zero as the distance increases. Figure~\ref{fig:spatial}(d) illustrates the distribution of absolute residuals $|y_i - \hat\theta_i|$.

\section{Discussion}
The optimality and rapid convergence guarantee of the DC algorithm makes it a compelling alternative or warm-start for MCMC algorithms. For example, for high-dimensional problems with large $n$ and $p$, the computation per each MCMC iteration can be intensive \citep{nishimura2023prior}, hence it is beneficial to use optimization to shorten the burn-in period. In this article, we focus on simple linear regression model due to the tractability of integrating out the regression coefficients and noise variance. For generalized linear models such as logistic or probit regression, one may use data augmentation to obtain a conditional normal distribution for the regression coefficients, hence allowing marginalization; however, the expectation of the function of the augmented variable (such as the quadratic form of the truncated normal variable in probit regression) is not fully tractable, hence some approximation such as Monte Carlo E-step is needed. This makes the optimization algorithm fundamentally more difficult to analyze theoretically. Alternatively, one could consider using a variational distribution or exploring alternative minorization steps in place of the Monte Carlo E-step, which would be an interesting direction for future research.

\noindent\textbf{Disclosure of the use of generative AI:} In preparing this paper, the author used
Claude Code to assist with code development and figure plotting. The author takes full
responsibility for all content presented.

\bibliographystyle{chicago}
\bibliography{ref}

\appendix

\section{Proof}

\subsection{Proof of Theorem \ref{thm:diverging}}
\begin{proof}
    Let
    $
    A:=X^\top X + D.
    $
    Since $X=[x_j,X_{-j}]$ and $D=\diag(d_j,D_{-j})$, we may write
    $$
    A=
    \begin{pmatrix}
    x_j^\top x_j+d_j & x_j^\top X_{-j} \\
    X_{-j}^\top x_j & X_{-j}^\top X_{-j}+D_{-j}
    \end{pmatrix}.
    $$
    By the Schur complement formula,
    $
    [A^{-1}]_{jj}
    =
    \{
    x_j^\top x_j+d_j
    -
    x_j^\top X_{-j}(X_{-j}^\top X_{-j}+D_{-j})^{-1}X_{-j}^\top x_j
    \}^{-1}=\{d_j+x_j^\top P_j(D_{-j})x_j\}^{-1}.
    $
    Similarly,
    $$
    [A^{-1}X^\top y]_j
    =
    \frac{
    x_j^\top y
    -
    x_j^\top X_{-j}(X_{-j}^\top X_{-j}+D_{-j})^{-1}X_{-j}^\top y
    }{
    d_j+x_j^\top P_j(D_{-j})x_j
    }
    =
    \frac{x_j^\top P_j(D_{-j})y}{d_j+x_j^\top P_j(D_{-j})x_j}.
    $$
    
    Differentiating the loss with respect to $d_j$ gives
    $$
    \frac{\partial \text{loss}(D)}{\partial d_j}
    =
    -\frac{1}{2d_j}
    +\frac12 [A^{-1}]_{jj}
    +
    \frac{a_0+n/2}{2\left\{b_0+\frac12\text{RSS}(D)\right\}}
    \Bigl([A^{-1}X^\top y]_j\Bigr)^2.
    $$
    Substituting the two identities above yields
    $$
    \frac{\partial \text{loss}(D)}{\partial d_j}
    =
    -\frac{x_j^\top P_j(D_{-j})x_j}
    {2\,d_j\bigl(d_j+x_j^\top P_j(D_{-j})x_j\bigr)}
    +
    \frac{a_0+n/2}
    {2\left\{
    b_0+\frac12\text{RSS}(D)
    \right\}}
    \cdot
    \frac{\bigl(x_j^\top P_j(D_{-j})y\bigr)^2}
    {\bigl(d_j+x_j^\top P_j(D_{-j})x_j\bigr)^2}.
    $$
    
    Next, using the Woodbury identity
    $
    I_n-X(X^\top X+D)^{-1}X^\top
    =
    \left(I_n+XD^{-1}X^\top\right)^{-1},
    $
    together with the partition $X=[x_j,X_{-j}]$, we obtain
    $$
    \text{RSS}(D)
    =
    y^\top P_j(D_{-j})y
    -
    \frac{\bigl(x_j^\top P_j(D_{-j})y\bigr)^2}
    {d_j+x_j^\top P_j(D_{-j})x_j}.
    $$
    Now on the inequality
    $$
    \frac{\partial \text{loss}(D)}{\partial d_j}<0,
    $$
    substituting the above expression for $\text{RSS}(D)$ and simplifying gives the equivalent condition
    $$
    \bigl(x_j^\top P_j(D_{-j})y\bigr)^2
    <
    \frac{
    \left\{
    b_0+\frac12\,y^\top P_j(D_{-j})y
    \right\}
    \,x_j^\top P_j(D_{-j})x_j\,
    \bigl(d_j+x_j^\top P_j(D_{-j})x_j\bigr)
    }{
    (a_0+n/2)\,d_j+\frac12\,x_j^\top P_j(D_{-j})x_j
    }.
    $$
    
    For fixed $D_{-j}$, define
    $$
    R(d_j)
    :=
    \frac{
    \left\{
    b_0+\frac12\,y^\top P_j(D_{-j})y
    \right\}
    \,x_j^\top P_j(D_{-j})x_j\,
    \bigl(d_j+x_j^\top P_j(D_{-j})x_j\bigr)
    }{
    (a_0+n/2)\,d_j+\frac12\,x_j^\top P_j(D_{-j})x_j
    }.
    $$
    A direct differentiation shows
    $$
    R'(d_j)
    =
    \frac{
    \left\{
    b_0+\frac12\,y^\top P_j(D_{-j})y
    \right\}
    \bigl(x_j^\top P_j(D_{-j})x_j\bigr)^2
    \left(\frac12-\left(a_0+\frac n2\right)\right)
    }{
    \left\{
    (a_0+n/2)\,d_j+\frac12\,x_j^\top P_j(D_{-j})x_j
    \right\}^2
    }
    <0,
    $$
    since $a_0+n/2>\frac12$. Hence $R(d_j)$ is strictly decreasing in $d_j>0$, and therefore
    $$
    \inf_{d_j>0}R(d_j)
    =
    \lim_{d_j\to\infty}R(d_j)
    =
    \frac{
    \left\{
    b_0+\frac12\,y^\top P_j(D_{-j})y
    \right\}
    \,x_j^\top P_j(D_{-j})x_j
    }{
    a_0+n/2
    }.
    $$
    It follows that if
    $$
    \bigl(x_j^\top P_j(D_{-j})y\bigr)^2
    <
    \frac{
    \left\{
    b_0+\frac12\,y^\top P_j(D_{-j})y
    \right\}
    \,x_j^\top P_j(D_{-j})x_j
    }{
    a_0+n/2
    },
    $$
    then 
    $
    {\partial \text{loss}(D)}/{\partial d_j}<0
    \qquad
    \text{for all } d_j>0.
    $
    \end{proof}

\subsection{Proof of Theorem \ref{thm:global_nonunique}}

\begin{proof}
     In this proof, we use $Q=X^\top X$ and $a=X^\top y$, $c=a_0+\frac n2$ to simplify the notation, and
    $
    r(D)= b_0+\frac12\Bigl(y^\top y-a^\top(Q+D)^{-1}a\Bigr).
    $

    Both $\tilde g$ and $h$ are closed, convex, and continuously differentiable on $\mathbb D$.
    The constraint set $\mathcal D$ is bounded, closed, and convex with
    $\mathcal D\subseteq \dom \tilde g \cap \dom h$.
    Since $\tilde g$ is strictly convex, the DCA subproblem \eqref{eq:dca_subproblem}
    is convex and has a unique solution for any $D^{(t)}\in\mathcal D$. Since $Q+D\succeq a I$ for all $D\in\mathcal D$, the function $h(D)$ is finite and continuous on $\mathcal D$.
    Together with the continuity and lower boundedness of $\tilde g$ on $\mathcal D$,
    the loss function $f=\tilde g-h$ is continuous on the compact set $\mathcal D$.
    Hence $f$ attains its minimum on $\mathcal D$, so $\mathcal S\neq\varnothing$ and
    $f^\star>-\infty$.

    We now show the existence of $\mu>0$ such that
    \(
    \mu\bigl(f(D)-f^\star\bigr)
    \le
    d_{h^*}\!\bigl(\nabla \tilde g(D)+U,\nabla h(D)\bigr),
    \qquad
    \forall\, D\in\mathcal D,\;
    U\in N_{\mathcal D}(D),\;
    \|U\|_F\le r,
    \)
    where $d_{h^*}(U,V)$ is the Bregman divergence of $h^*$ at $V$ with respect to $U$; $h^*$ is the convex conjugate of $h$.

    Since $\mathcal D$ is compact and $d_j\ge a>0$ on $\mathcal D$,
    the Hessian of $\tilde g$ is bounded above on $\mathcal D$.
    Also, since $Q+D\succeq a I$ for all $D\in\mathcal D$ and $r(D)\ge b_0>0$,
    the Hessian of $h$ is bounded above on $\mathcal D$.
    Therefore, $h$ is $L_h$-smooth on $\mathcal D$ for some $L_h>0$.

    Since $h$ is convex and $L_h$-smooth, its convex conjugate $h^*$ is
    $L_h^{-1}$-strongly convex. 
    \(
    h^*(V_1) \ge h^*(V_2)+ \langle V_1-V_2, \nabla h^*(V_2)\rangle + \frac{1}{2L_h}\|V_1-V_2\|_F^2
    \)
    
    Equivalently, the Bregman divergence $d_{h^*}(V_1,V_2)$
    \begin{equation}\label{eq:bregman_lower_global}
    d_{h^*}(V_1,V_2)\ge \frac{1}{2L_h}\|V_1-V_2\|_F^2
    \qquad
    \forall\, V_1,V_2.
    \end{equation}
    Applying \eqref{eq:bregman_lower_global} with
    $
    V_1=\nabla \tilde g(D)+U,
    \;
    V_2=\nabla h(D),
    $
    yields
    \begin{equation}\label{eq:bregman_residual_global}
    d_{h^*}\!\bigl(\nabla \tilde g(D)+U,\nabla h(D)\bigr)
    \ge
    \frac{1}{2L_h}\|\nabla \tilde g(D)+U-\nabla h(D)\|_F^2.
    \end{equation}

    \noindent\textbf{1. Neighborhood of the solution set:}
    We first consider the neighborhood of the solution set
    \(
    \mathcal N_\rho
    :=
    \bigl\{
    D\in\mathcal D:\ \operatorname{dist}(D,\mathcal S)\le \rho
    \bigr\}.
    \)
    Let $D\in\mathcal N_\rho$ and let $U\in N_{\mathcal D}(D)$ satisfy $\|U\|_F\le r$.
    By assumptions (A1) and (A2),
    \(
    \|\nabla \tilde g(D)+U-\nabla h(D)\|_F^2
    \ge
    m^2\,\operatorname{dist}(D,\mathcal S)^2
    \ge
    \frac{m^2}{M}\bigl(f(D)-f^\star\bigr).
    \)
    Combining this with \eqref{eq:bregman_residual_global} gives
    \begin{equation}\label{eq:dcpl_local_set}
    d_{h^*}\!\bigl(\nabla \tilde g(D)+U,\nabla h(D)\bigr)
    \ge
    \frac{m^2}{2L_h M}\bigl(f(D)-f^\star\bigr),
    \qquad
    \forall\, D\in\mathcal N_\rho,\;
    U\in N_{\mathcal D}(D),\;
    \|U\|_F\le r.
    \end{equation}

    \noindent\textbf{2. Neighborhood away from the solution set:} It remains to control the region away from $\mathcal S$.
    Define
    \(
    K
    :=
    \bigl\{
    D\in\mathcal D:\ \operatorname{dist}(D,\mathcal S)\ge \rho
    \bigr\}.
    \)
    If $K=\varnothing$, then \eqref{eq:dcpl_local_set} already proves the claim globally.
    Assume now that $K\neq\varnothing$.

    Consider the set
    \(
    \Xi
    :=
    \bigl\{
    (D,U):\ D\in K,\ U\in N_{\mathcal D}(D),\ \|U\|_F\le r
    \bigr\}.
    \)

    Because $\mathcal D$ is compact, $\|U\|_F\le r$, and the graph of the normal cone mapping
    $N_{\mathcal D}$ is closed, the set $\Xi$ is compact.
    The map
    \(
    (D,U)\mapsto \|\nabla \tilde g(D)+U-\nabla h(D)\|_F
    \)
    is continuous on $\Xi$, so it attains its minimum there.
    By the assumption (A3),
    this minimum cannot be zero on $K$. Hence there exists $\delta_\Xi>0$ such that
    $
    \|\nabla \tilde g(D)+U-\nabla h(D)\|_F
    \ge
    \delta_\Xi,
    \;
    \forall\, (D,U)\in\Xi.
    $
    Also, since $f$ is continuous on the compact set $K$,
    $
    G_\Xi:=\max_{D:(D,U)\in\Xi}\bigl(f(D)-f^\star\bigr)<\infty.
    $
    Therefore, for every $(D,U)\in\Xi$,
    $
    \|\nabla \tilde g(D)+U-\nabla h(D)\|_F^2
    \ge
    \delta_\Xi^2
    \ge
    \frac{\delta_\Xi^2}{G_\Xi}\bigl(f(D)-f^\star\bigr).
    $
    Combining this with \eqref{eq:bregman_residual_global}, we get
    \begin{equation}\label{eq:dcpl_away_set}
    d_{h^*}\!\bigl(\nabla \tilde g(D)+U,\nabla h(D)\bigr)
    \ge
    \frac{\delta_\Xi^2}{2L_h G_\Xi}\bigl(f(D)-f^\star\bigr),
    \qquad
    \forall\, D\in K,\;
    U\in N_{\mathcal D}(D),\;
    \|U\|_F\le r.
    \end{equation}

    Finally, define
    \[
    \mu
    :=
    \min\!\left\{
    \frac{m^2}{2L_h M},
    \frac{\delta_\Xi^2}{2L_h G_\Xi}
    \right\}>0,
    \]
    with the convention that the second term is omitted if $K=\varnothing$.
    Then \eqref{eq:dcpl_local_set} and \eqref{eq:dcpl_away_set} together imply
    $
    \mu\bigl(f(D)-f^\star\bigr)
    \le
    d_{h^*}\!\bigl(\nabla \tilde g(D)+U,\nabla h(D)\bigr),
    \;
    \forall\, D\in\mathcal D,\;
    U\in N_{\mathcal D}(D),\;
    \|U\|_F\le r.
    $

    By first-order optimality of the DC subproblem, there exists
    $U^{(t+1)}\in N_{\mathcal D}(D^{(t+1)})$ such that
    \(
    \nabla \tilde g(D^{(t+1)}) + U^{(t+1)} = \nabla h(D^{(t)}).
    \)
    For $t\ge t_0$, applying the inequality at $D=D^{(t+1)}$ and $U=U^{(t+1)}$ gives
    \(
    \mu\bigl(f(D^{(t+1)})-f^\star\bigr)
    \le
    d_{h^*}\!\bigl(\nabla \tilde g(D^{(t+1)})+U^{(t+1)},\nabla h(D^{(t+1)})\bigr)
    =
    d_{h^*}\!\bigl(\nabla h(D^{(t)}),\nabla h(D^{(t+1)})\bigr).
    \)
    By the standard conjugacy identity for Bregman divergences,
    \(
    & d_{h^*}\!\bigl(\nabla h(D^{(t)}),\nabla h(D^{(t+1)})\bigr)
    =
    d_h(D^{(t+1)},D^{(t)})\\
      &=
      h(D^{(t+1)})-h(D^{(t)})
      -\langle \nabla h(D^{(t)}),D^{(t+1)}-D^{(t)}\rangle \\
      &=
      f(D^{(t)})-f(D^{(t+1)})
      -
      \Bigl(
      \tilde g(D^{(t)})-\tilde g(D^{(t+1)})
      -\langle \nabla h(D^{(t)}),D^{(t)}-D^{(t+1)}\rangle
      \Bigr).
\)
      Note that
  \begin{align*}
  &\tilde g(D^{(t)})-\tilde g(D^{(t+1)})
  -\langle \nabla h(D^{(t)}),D^{(t)}-D^{(t+1)}\rangle \\
  &\qquad=
  \tilde g(D^{(t)})-\tilde g(D^{(t+1)})
  -\langle \nabla g(D^{(t+1)})+U^{(t+1)},D^{(t)}-D^{(t+1)}\rangle \\
  &\qquad=
  d_{\tilde g}(D^{(t)},D^{(t+1)})
  -
  \langle U^{(t+1)},D^{(t)}-D^{(t+1)}\rangle \\
  &\qquad\ge
  d_{\tilde g}(D^{(t)},D^{(t+1)}) \ge 0, 
  \end{align*}
  since $U^{(t+1)}\in N_{\mathcal D}(D^{(t+1)})$ and $D^{(t)}\in\mathcal D$ imply
  $
  \langle U^{(t+1)},D^{(t)}-D^{(t+1)}\rangle\le 0.
  $
  Therefore,
  \begin{align*}
      d_h(D^{(t+1)},D^{(t)})
      &\le
      f(D^{(t)})-f(D^{(t+1)}).
  \end{align*}
  
    Therefore,
    $
    \mu\bigl(f(D^{(t+1)})-f^\star\bigr)
    \le
    f(D^{(t)})-f(D^{(t+1)}).
    $
    Rearranging yields
    \(
    f(D^{(t+1)})-f^\star
    \le
    \frac{1}{1+\mu}\bigl(f(D^{(t)})-f^\star\bigr).
    \)
    This proves the claimed linear convergence rate in loss function value.

    We see that as $r$ decreases or $\rho$ increases, $\Xi$ becomes smaller, hence $\delta^2_\Xi/G_\Xi$ becomes larger.
\end{proof}

\begin{remark}
  Our proof is inspired by a recent work of \cite{yao2023globally}. There are two fundamental differences between the proof in \cite{yao2023globally} and ours: (1) the lower bound of a Bregman divergence in Theorem 1 of \cite{yao2023globally} was stated as an assumption, we formally prove it for our problem;
  (2) the proof of Theorem 1 in \cite{yao2023globally} further assumes the bound constant $r$ can increase arbitrarily, we do not make this assumption but make $r$ an upper bound of the norm of the DC multipliers after sufficient iterations.
\end{remark}

\end{document}